\documentclass[12pt,a4paper]{article}

\usepackage{amsmath}
\usepackage{amssymb}
\usepackage{amsfonts}
\usepackage{amsthm}
\usepackage{mathtools}
\usepackage{graphicx}
\usepackage{float}
\usepackage{enumitem}
\usepackage{booktabs}
\usepackage{multirow}
\usepackage{array}
\usepackage{color}
\usepackage{hyperref}
\usepackage{geometry}
\allowdisplaybreaks

\DeclareMathOperator{\argmax}{arg\,max}

\newcommand{\E}{\mathbb{E}}

\newcommand{\TV}{\mathrm{TV}}
\newcommand{\psib}{\psi_{\mbox{\tiny b}}(\beta,R)}

\newtheorem{theorem}{Theorem}[section]
\newtheorem{lemma}[theorem]{Lemma}

\newtheorem{definition}[theorem]{Definition}

\title{A Statistical-Physics Refinement of Soft Covering}

\author{Neri Merhav\\[6pt]
\small The Viterbi Faculty of Electrical and Computer Engineering,\\
\small Technion---Israel Institute of Technology, Technion City,\\
\small Haifa 3200003, Israel\\
\small \texttt{merhav@ee.technion.ac.il}}

\date{}

\begin{document}

\maketitle

\begin{abstract}
We study the channel output distribution induced by a
random code of rate $R$,
from the perspective of statistical 
physics. The central object is the partition function 
$Z_n(\beta|\mathcal{C}) = \sum_{y^n}[P_{Y^n|\mathcal{C}}(y^n)]^\beta$, 
where $y^n$ is the channel output vector, $\mathcal{C}$ is the code, and $\beta > 0$ plays the role of inverse temperature. 
More precisely, our focus is on the associated \emph{annealed free energy},
$\psi(\beta,R) = \lim_{n\to\infty}\frac{1}{n}\log\E[Z_n(\beta|\mathcal{C})]$, 
where the expectation is with respect to the randomness of $\mathcal{C}$.
This quantity encodes the full R\'{e}nyi spectrum of the output distribution. 
The single-letter formula derived for the annealed free energy
decomposes into two branches which reflect a ``competition'' between two populations 
of codewords. One is the \emph{bulk branch}, $\psi_{\mbox{\tiny b}}(\beta,R)$, 
which is driven by typical codewords
and the other one is the \emph{sparse branch}
$\psi_{\mbox{\tiny s}}(\beta,R)$, which is driven 
by a-typical codewords, where the qualifiers `typical' and `atypical' are in a
sense to become apparent later.
We analyze the phase structure of each branch separately and characterize 
their competition. Both branches are derived for all $\beta > 0$. The phase 
boundary $R^\star(\beta)$, where the two branches are equal, is analyzed for 
$\beta \geq 1$, where it has an explicit closed-form expression. 
The phase diagram in the first quadrant of the $(\beta, R)$ plane has four regions separated 
by three boundaries: $R = I^{\mbox{\tiny b}}(\beta)$ (bulk branch transition), 
$R = R^\star(\beta)$ (bulk--sparse competition boundary), 
and $R = I^{\mbox{\tiny s}}(\beta)$ (sparse branch transition), all meeting at 
the point $(\beta, R) = (1, I(X;Y))$, where $I(X;Y)$ is the mutual information
induced by the input type and the channel. Applications to guesswork, channel resolvability, 
and hypothesis testing are discussed, and all results are 
illustrated with a numerical example of a Z-channel.
\end{abstract}


\tableofcontents
\newpage

\setcounter{section}{0}

\section{Introduction}
\label{intro}

Given a codebook $\mathcal{C} = \{x^{n}(m)\}_{m=1}^{M}$ of size $M = e^{nR}$ and a 
discrete memoryless channel (DMC), $\{W(y|x)\}$, the induced channel output
distribution is given by the mixture

\begin{equation}
  P_{Y^{n}}(y^{n}) = \frac{1}{M} \sum_{m=1}^{M} W^{n}(y^{n}|x^{n}(m)).
  \label{eq:output-dist}
\end{equation}
This object is central to channel resolvability~\cite{han1993} and the soft-covering problem
(see, e.g., \cite{cuff16}, \cite{hayashi2006}, \cite{hou2014} and references
therein). Classical results, in this context, focus on a single threshold: 
if $R > I(X;Y)$, $I(X;Y)$ being the mutual information induced by the input distribution and
the channel $W$,
then the distribution $P_{Y^{n}}$ converges to the i.i.d.\ product law
$P_{Y}^{\otimes n}$ in total variation as well as in some other metrics between
probability distributions.

This description, however, captures average behavior only. It says nothing about the internal
geometry of $P_{Y^{n}}$: how the probability mass is distributed across output sequences, how many codewords
support a typical and an atypical output, 
or how the self-information, $-\frac{1}{n}\log P_{Y^{n}}(y^{n})$, fluctuates. These questions
require going along a deeper journey, beyond the Shannon entropy and mutual information.

In this work, we study $P_{Y^{n}}$ through the function,

\begin{equation}
  Z_{n}(\beta) = \sum_{y^{n}} [P_{Y^{n}}(y^{n})]^{\beta}=\sum_{y^{n}}
\exp\{-\beta\log[1/P_{Y^n}(y^n)]\}, \quad \beta > 0,
  \label{eq:partition}
\end{equation}
where the second representation is readily recognized as a \emph{partition
function}, with $\beta > 0$ playing the role of \emph{inverse temperature},
and $\log[1/P_{Y^n}(y^n)]=-\log P_{Y^{n}}(y^{n})$ being the \emph{energy
function} (a.k.a.\ the \emph{Hamiltonian}) associated with every
\emph{micro-state} $y^{n}$. In other words, this is identified as the canonical partition function 
of a statistical-mechanical system whose micro-states are channel output sequences.
Two special values bracket the range: $Z_n(1) = 1$ (normalization, $\beta=1$) 
and $\lim_{\beta\to\infty}[Z_n(\beta)]^{1/\beta} =\max_{y^n}P_{Y^n}(y^n)$ (the
mode, or the ground state in the physics jargon). 
In general, $\frac{1}{1-\beta}\log Z_n(\beta)$ is exactly the definition of $H_\beta(P_{Y^n})$, 
the R\'{e}nyi entropy of order $\beta \neq 1$, pertaining to the output distribution $P_{Y^n}$, 
which is non-increasing in $\beta$.
The associated \emph{free energy} $\psi(\beta,R)$ encodes the exponential
growth/decay rate of $Z_n(\beta)$
as a function of $n$.
Since $Z_n(\beta)$ depends on the random codebook, 
and actually, should be denoted $Z_n(\beta|\mathcal{C})$,
we study the free energy 
behavior for the average code, which is identified with the \emph{annealed free energy},

\begin{equation}
  \psi(\beta,R) = \lim_{n\to\infty}\frac{1}{n}\log\E\{Z_n(\beta|\mathcal{C})\},
\end{equation}
provided that the limit exists. Here,
$\E\{\cdot\}$ denotes the expectation operator with respect to (w.r.t.) the randomness of the
code $\mathcal{C}$.
Its phase structure --- the rich dependence on $\beta$ and $R$ --- is the main subject
of this work.

It should be emphasized that many earlier
statistical-mechanical analyses of random coding, were based on Derrida's random
energy model (REM) (see, e.g.,
\cite{MM09}, \cite{merhav2010} and references therein). These are based on partition functions
whose micro-states were the channel inputs (i.e., the codewords) for a fixed
channel output
sequence $y^{n}$, that is,

\begin{equation}
  Z_n^{\mathrm{input}}(\beta) = \sum_{m=1}^{M} 
[W^{n}(y^{n}|x^{n}(m))]^\beta,
  \label{eq:zinput}
\end{equation}
for fixed $y^n$, whereas, here the micro-states are the channel outputs, as mentioned earlier.
This difference is significant because here,
each energy term, $-\log P_{Y^{n}}(y^{n})$, is itself a log-partition function over codewords. This two-level structure
makes the problem considerably harder and richer than that of an ordinary REM.

The main contributions of this work are as follows.
We derive the annealed free energy and analyze the phase structure 
of its two branches (Theorems~\ref{thm:annealed} and~\ref{thm:phase}).
The free energy decomposes into a \emph{bulk branch} and a \emph{sparse branch},
reflecting a competition between two populations of codewords.
The bulk branch is driven by \emph{typical} codewords
and the sparse branch is driven by \emph{atypical} codewords, where the
meanings of the qualifiers `typical' and `atypical' will become apparent in
the sequel.
Both branches are derived for all $\beta > 0$ and analyzed separately.
The bulk branch, $\psi_{\mbox{\tiny b}}(\beta,R)$ (Section~\ref{sec:annealed}),
has a single phase transition at
a critical rate $R = I^{\mbox{\tiny b}}(\beta)$---the mutual information of the bulk optimizer
(the unconstrained optimizer of the bulk branch)---where a rate constraint
changes from inactive to active. A similar comment applies to the sparse
branch, $\psi_{\mbox{\tiny s}}(\beta,R)$, whose phase boundary is $R=I^{\mbox{\tiny s}}(\beta)$.
The total annealed free energy $\psi(\beta,R)$ has 
a phase boundary $R^\star(\beta)$ given by an explicit closed-form formula
(Theorem~\ref{thm:phase}(b), Section~\ref{sec:annealed}).
Together, the three boundaries divide the $(\beta, R)$ quadrant $\{R\ge 0,~\beta \geq 1\}$
into four regions, denoted, A, B, C, and D.
By construction, $\psi(\beta,R) \geq \psi_{\mbox{\tiny b}}(\beta,R)$ always,
with equality in region~A and
strict inequality in regions~B, C, and D, where the sparse branch dominates.
For the sparse branch, a fully explicit closed form formula, in terms of channel transition
probabilities and $\beta$ alone, holds whenever the rate is below
$I^{\mbox{\tiny s}}(\beta)$
(see eq.\ (\ref{eq:sparse-closed})).
All results are illustrated in a numerical example in
Section~\ref{sec:zchannel}.

A few words are in order with regard to earlier related work and the differences relative to the
present work. The model~\eqref{eq:output-dist} is the canonical object of channel resolvability 
and the soft-covering problem. Han and Verd\'{u}~\cite{han1993} established the 
resolvability threshold $R = I(X;Y)$ under total variation, 
$\|P_{Y^{n}} - P_{Y}^{\otimes n}\|_{\TV} \to 0$ for $R > I(X;Y)$, 
and showed the threshold is tight. Hayashi~\cite{hayashi2006} derived exponential convergence rates 
under total variation and normalized relative entropy, a.k.a.\ the
Kullback-Leibler (KL) divergence. 
Hou and Kramer~\cite{hou2014} extended the analysis to the 
R\'{e}nyi divergence $D_{\alpha}(P_{Y^{n}} \| P_{Y}^{\otimes n})$, 
showing the critical rate remains $I(X;Y)$ for all $\alpha \in (0,\infty)$ 
but the exponents differ. 
Yu and Tan~\cite{yu2020} derived exact error and strong-converse exponents 
for the soft-covering problem under KL divergence and total variation.
In a slightly earlier but closely related and highly relevant work
\cite{yu2019}, the same
authors characterized the \emph{R\'{e}nyi resolvability}---the minimum rate $R$ 
required for the R\'{e}nyi divergence $D_\alpha(P_{Y^n|\mathcal{C}}\|P_Y^{\otimes n})$ 
to vanish asymptotically. They showed that for $\alpha \leq 1$ the 
threshold remains $I(X;Y)$, while for $\alpha > 1$ it is strictly larger than $I(X;Y)$ and depends on $\alpha$. 

Our work is complementary to~\cite{yu2019}: rather than characterizing the \emph{threshold rate} 
at which a divergence to $P_Y^{\otimes n}$ vanishes, we study the exact value of the 
annealed free energy at every rate $R$ and inverse temperature $\beta$, revealing 
a two-branch phase diagram with three distinct phase boundaries. We analyze the \emph{below-threshold} 
regime with the same precision as the above-threshold regime, 
and expose qualitative phenomena---condensation, sparse-branch dominance, 
and the annealed/quenched gap---invisible to divergence-based analyses. 
In particular, our bulk branch boundary, $R = I^{\mbox{\tiny b}}(\beta)$, for $\beta > 1$ 
is precisely the R\'{e}nyi resolvability threshold of~\cite{yu2019}, 
now seen as one of three phase boundaries in a richer thermodynamic landscape.

In general, all earlier results, in this context, share a common feature: 
they measure how close $P_{Y^{n}}$ is to a \emph{target} distribution $P_{Y}^{\otimes n}$ 
via some metric, and they all identify $R = I(X;Y)$ 
as the single relevant threshold. The present paper asks a fundamentally different and complementary question: not how close $P_{Y^{n}}$ is to a target, but what is the \emph{internal geometry} of $P_{Y^{n}}$ itself?

Our partition function $Z_{n}(\beta)$ encodes the R\'{e}nyi structure of $P_{Y^{n}}$ 
without reference to any external target distribution, and reveals phenomena invisible 
to total variation or KL divergence. It turns out that the classical
threshold, $R = I(X;Y)$, is only the \emph{beginning} of the story: 
the $(\beta, R)$ quadrant $\{\beta \geq 1,~R\ge 0\}$ is divided into \emph{four} 
distinct regions by three phase boundaries, all passing through the point $(\beta, R) = (1, I(X;Y))$. Concretely:
\begin{itemize}
  \item Even for $R > I(X;Y)$ (where $\|P_{Y^{n}}-P_Y^{\otimes n}\|_{\TV}\to 0$),
    the annealed free energy reveals a \emph{condensed phase}: when $R < I^{b}(\beta)$
    for some $\beta>1$, the bulk branch of $\psi$ has an active rate
    constraint, signalling that probability mass concentrates on outputs supported
    by sub-exponentially few codewords --- output condensation invisible to total
    variation distance.
  \item The annealed free energy has a phase boundary $R = R^\star(\beta)$
    where the sparse branch takes over from the bulk branch, signalling that
    the atypical codewords dominate the ensemble average.
  \item Since $\log Z_n(\beta|\mathcal{C}) = (1-\beta)\,H_\beta(P_{Y^n|\mathcal{C}})$
    exactly for every fixed codebook (by the definition of R\'{e}nyi entropy),
    the bulk branch $\psi_{\mbox{\tiny b}}(\beta,R)$ encodes the R\'{e}nyi entropy rate
    $H_\beta(P_{Y^n|\mathcal{C}})/n$ of the output mixture for a typical codebook
    (for $\beta\geq 1$ and $R>R^\star(\beta)$, this is proved in
    Theorem~\ref{thm:self-averaging}; see also Appendix).
    The full annealed free energy thus encodes the complete R\'{e}nyi-order profile of the
    output distribution, of which classical soft-covering ($\beta\to 1$) is a single cross-section.
\end{itemize}
In summary, this work does not contradict soft-covering results---it refines 
and extends them by revealing the fine structure of $P_{Y^{n}}$ that 
lies beneath the single threshold $R = I(X;Y)$.
These structural findings have direct operational consequences, which we now
describe.\\

\noindent
1.~{\em Guesswork and R\'{e}nyi entropy.}
Arikan~\cite{arikan1996} showed that guesswork moments $\E[G(Y^{n})^{s}]$ grow 
as $e^{nsH_{1/(1+s)}(P_{Y^{n}})}$. The annealed free energy encodes the R\'{e}nyi entropy rate of 
the output mixture, extending Ar{\i}kan's analysis to the case where the source is itself 
a random-coding mixture. The phase structure of 
$\psi(\beta,R)$ produces regime changes in guessing complexity absent in the
i.i.d.\ case.\\

\noindent
2.~{\em Hypothesis testing.}
The Chernoff exponent~\cite{blahut1974} for testing $P_{Y^n|\mathcal{C}}$
against $P_Y^{\otimes n}$ equals
$\xi(R)=\max_{0\leq\beta\leq 1}[(1-\beta)\log|\mathcal{Y}|-\psi(\beta,R)]$,
directly connecting our free energy to the optimal test.
With $P_Y$ being uniform, $\xi(R)=0$ for $R\geq I(X;Y)$,
meaning the two distributions are not exponentially distinguishable at and above
the soft-covering threshold.
The phase structure of $\psi(\beta,R)$ produces qualitative regime changes in $\xi(R)$
as $R$ varies, and connects to the R\'{e}nyi resolvability results of
Yu and Tan~\cite{yu2019}; see Section~\ref{sec:applications}.\\

\noindent
3.~{\em Statistical mechanics of codes.}
Sourlas~\cite{sourlas1989} established the connection between linear codes and spin-glass models. 
Montanari~\cite{montanari2001} analyzed the phase transition in turbo codes. 
M\'ezard and Montanari \cite{MM09} (Chapters 5 and 6) and
Merhav~\cite{merhav2010} provided a comprehensive treatment of random coding 
via statistical physics, with the partition function summing over codewords (inputs) 
for a fixed received sequence. The present work complements~\cite{merhav2010} 
by placing the partition function over outputs: the resulting two-level 
hierarchical model has a distinct phase structure governed by the code rate $R$ rather than by signal-to-noise ratio.

The outline of the remaining part of this article is as follows.
Section~\ref{sec:model} defines the model and establishes the notation
conventions.
Section~\ref{sec:annealed} derives the annealed free energy and analyzes the phase
structure of its bulk and sparse branches; the phase boundary $R^*(\beta)$
has an explicit closed-form formula for $\beta\geq 1$.
Section~\ref{sec:zchannel} presents the phase diagram (four regions, $\beta \geq 1$).
Section~\ref{sec:applications} discusses applications and implications of our
results. Finally,
Section \ref{sec:conclusion} summarizes the paper and provides an outlook.

\section{Model, Definitions and Notation}
\label{sec:model}

Throughout the paper, $n$ denotes the block length. 
An $n$-vector over alphabet $\mathcal{X}$ is displayed as 
$x^{n} = (x_1, x_2, \ldots, x_n) \in \mathcal{X}^{n}$, 
$\mathcal{X}$ being the single-letter finite alphabet,
and similarly $y^{n} \in \mathcal{Y}^{n}$, where the single-letter alphabet
$\mathcal{Y}$ is also finite.
The empirical distribution (type) of $x^{n}$ is denoted 
$\hat{P}_{x^{n}}(a) = \frac{1}{n}\#\{i:~x_i = a\}$ for all $a \in \mathcal{X}$. 
A similar definition applies to joint types of pairs of sequences,
$(x^n,y^n)$. A DMC $W$ is a stochastic matrix $W:\mathcal{X}\to\mathcal{Y}$
and for sequences $x^{n}$, $y^{n}$, 
we write $W^{n}(y^{n}|x^{n}) = \prod_{i=1}^{n} W(y_i|x_i)$.
Throughout the sequel, all logarithms are natural (base $e$); 
Accordingly, information measures are given in nats. 
We use the standard exponential-equivalence notation $f(n) \doteq g(n)$ to mean 
$\lim_{n\to\infty} \frac{1}{n}\log\frac{f(n)}{g(n)}=0$, 
i.e.\ $f$ and $g$ have the same exponential growth/decay rate. 
The notation of information measures is as follows.
$D(P\|Q)$ the Kullback-Leibler (KL) divergence, 

\begin{equation}
D(P\|Q)=\sum_{x\in\mathcal{X}}P(x)\log\frac{P(x)}{Q(x)}.
\end{equation}
Throughout the paper, $Q_{XY}$ denotes a joint distribution on $\mathcal{X}\times\mathcal{Y}$
with $X$-marginal given by $Q_X=P_X$, $P_X$ being a fixed input distribution
on the finite input alphabet $\mathcal{X}$.
When $Q_{XY}$ needs to appear as a \emph{subscript} in an information measure,
we abbreviate $Q_{XY}$ by $Q$ to avoid cumbersome notation. Thus,
$H_Q(Y)$, $H_Q(Y|X)$, and $I_Q(X;Y)$ stand for
$H_{Q_{XY}}(Y)$, $H_{Q_{XY}}(Y|X)$, and $I_{Q_{XY}}(X;Y)$ respectively, which
are the entropy of $Y$, the conditional entropy of $Y$ given $X$, and the
mutual information between $X$ and $Y$, respectively, all induced by $Q_{XY}$. On the other
hand, when $Q_{XY}$ appears as an argument of a certain functional,
the full notation is used.
The notation $I(X;Y)$ (without subscript) 
refers to the mutual information induced by the fixed single-letter channel input
distribution $P_X$ and the channel transition probability matrix $W$, i.e.\ $I(X;Y) = I_{P_X W}(X;Y)$.
The KL divergence between a conditional distribution $Q_{Y|X}$ and $W$, with
weighting by $P_X$, is defined by

\begin{equation}
D(Q_{Y|X}\|W|P_X)=\sum_{x\in\mathcal{X}}P_X(x)\sum_{y\in\mathcal{Y}}Q_{Y|X}(y|x)\log\frac{Q_{Y|X}(y|x)}{W(y|x)}.
\end{equation}
A random codebook of rate $R > 0$ and block length $n$ is a 
collection $\mathcal{C} = \{x^{n}(m)\}_{m=1}^{M}$, $M = \lfloor e^{nR} \rfloor$, 
where the codewords $x^{n}(m)$ are drawn independently and uniformly at random from the type class

\begin{equation}
  \mathcal{T}(P_{X}) = \bigl\{ x^{n} \in \mathcal{X}^{n} : \hat{P}_{x^{n}}(a) = P_{X}(a)\ \forall\, a \in \mathcal{X} \bigr\},
\end{equation}
the set of all $n$-sequences with empirical distribution exactly $P_{X}$.
Given codebook $\mathcal{C}$, the induced output distribution is

\begin{equation}
  P_{Y^{n}|\mathcal{C}}(y^{n}) = \frac{1}{M} \sum_{m=1}^{M} W^{n}(y^{n}|x^{n}(m)).
\end{equation}
For $\beta > 0$, the partition function is defined as

\begin{equation}
  Z_{n}(\beta|\mathcal{C}) = \sum_{y^{n}} [P_{Y^{n}|\mathcal{C}}(y^{n})]^{\beta}.
  \label{eq:Zndef}
\end{equation}
The \emph{annealed free energy} is

\begin{equation}
  \psi(\beta, R) = \lim_{n\to\infty} \frac{1}{n} \log \E_{\mathcal{C}}\bigl[Z_{n}(\beta|\mathcal{C})\bigr],
  \label{eq:psi-ann}
\end{equation}
assuming that the limit exists.
For a joint type $Q_{XY}$ on $\mathcal{X} \times \mathcal{Y}$ with marginal $Q_{X} = P_{X}$, define:

\begin{equation}
  \ell(Q_{XY}) = H_{Q}(Y|X) + D(Q_{Y|X}\|W|P_{X}) \label{eq:ellQ}.
\end{equation}
Note that for $(x^n,y^n)$ of joint type $Q_{XY}$,
$W^n(y^n|x^n)=e^{-n\ell(Q_{XY})}$.
We shall often use the identity

\begin{equation}
  I_{Q}(X;Y) + \ell(Q_{XY}) = H_{Q}(Y) + D(Q_{Y|X}\|W|P_{X}). \label{eq:identity}
\end{equation}
Fix an output sequence $y^{n}$ of type $Q_{Y}$. The type-class enumerator
$N(Q_{XY}|y^n)$ counts how many codewords have a joint type (joint empirical
distribution) $Q_{XY}$ with $y^n$:

\begin{equation}
  N(Q_{XY}|y^{n}) = \#\bigl\{m:~(x^{n}(m), y^{n}) \in \mathcal{T}(Q_{XY})\bigr\},
\end{equation}
where $\mathcal{T}(Q_{XY})$ is the joint type class associated with the joint
distribution $Q_{XY}$.

Using the method of types~\cite{csiszar2011},
\cite{csiszar1998}, it is readily observed that
given $y^n$, each randomly selected codeword, $x^{n}(m)$, satisfies $(x^{n}(m), y^{n}) \in \mathcal{T}(Q_{XY})$ with probability

\begin{equation}
\frac{|\mathcal{T}(Q_{X|Y}|y^{n})|}{|\mathcal{T}(P_{X})|} \doteq \frac{e^{nH_{Q}(X|Y)}}{e^{nH_{Q}(X)}} = e^{-nI_{Q}(X;Y)},
\end{equation}
where $\mathcal{T}(Q_{X|Y}|y^{n}) = \{x^{n}:~(x^{n}, y^{n}) \in \mathcal{T}(Q_{XY})\}$ 
is the conditional type class, and where we have used the fact that $Q_{X} = P_{X}$. Since codewords
are drawn independently at random from $\mathcal{T}(P_X)$,
it is clear that $N(Q_{XY}|y^{n})$ is a binomial random variable with $M$ trials and
a probability of single success of the exponential order of $e^{-nI_{Q}(X;Y)}$.

\section{Annealed Free Energy and its Phase Structure}
\label{sec:annealed}

In this section, we derive a single-letter expression for the
annealed free energy and investigate its phase structure in the
$(\beta,R)$ plane. The first main result is the following.

\begin{theorem}[Annealed Free Energy, $\beta > 0$]
\label{thm:annealed}
For all $\beta > 0$:

\begin{equation}
  \psi(\beta, R) = \max\!\bigl\{\psi_{\mbox{\tiny b}}(\beta, R),\,
\psi_{\mbox{\tiny s}}(\beta, R)\bigr\},
\end{equation}
with

\begin{align}
  \psi_{\mbox{\tiny b}}(\beta, R) &= \sup_{\substack{Q_{XY}:\, Q_{X}=P_{X}\\ I_{Q}(X;Y)\leq R}} 
\bigl[H_Q(Y)-\beta(I_{Q}(X;Y)+\ell(Q_{XY})\bigr], \label{eq:psi-bulk}\\[6pt]
  \psi_{\mbox{\tiny s}}(\beta, R) &= R(1-\beta) + 
\sup_{\substack{Q_{XY}:\, Q_{X}=P_{X}\\ I_{Q}(X;Y)>R}}\{H_{Q}(Y|X) -
\beta\ell(Q_{XY})\}. \label{eq:psi-sparse}
\end{align}
\end{theorem}

For convenience, we define the functional

\begin{equation}
  F(Q_{XY}) := H_{Q}(Y|X) - \beta\ell(Q_{XY}),
  \label{eq:F-def}
\end{equation}
which is the common building block of both branches.
Using the identity, $H_Q(Y) = I_Q(X;Y) + H_Q(Y|X)$, we may present the two
branches of the annealed free energy as

\begin{align}
\psi_{\mbox{\tiny b}}(\beta, R) &= \sup_{\substack{Q_{XY}:\, Q_{X}=P_{X}\\ I_{Q}(X;Y)\leq R}}
\bigl[(1-\beta)I_{Q}(X;Y) +F(Q_{XY})\bigr]
  \label{eq:bulk-obj-F}\\
\psi_{\mbox{\tiny s}}(\beta, R) &= R(1-\beta) + 
\sup_{\substack{Q_{XY}:\, Q_{X}=P_{X}\\ I_{Q}(X;Y)>R}}F(Q_{XY}).
  \label{eq:sparse-obj-F}
\end{align}

\begin{proof}[Proof of Theorem~\ref{thm:annealed}]
We begin with the trivial identity,

\begin{equation}
  \E[Z_{n}(\beta)] = \E\left[\sum_{y^{n}} [P_{Y^{n}|\mathcal{C}}(y^{n})]^{\beta}\right] 
= \sum_{y^{n}}
\E\left\{\left[P_{Y^{n}|\mathcal{C}}(y^{n})\right]^{\beta}\right\}.
\end{equation}
We next focus on $\E[P_{Y^{n}|\mathcal{C}}(y^{n})^{\beta}]$ for a fixed $y^{n}$.
Define $S(y^{n})= \sum_{m=1}^{M} W^{n}(y^{n}|x^{n}(m))$, 
so that $P_{Y^{n}|\mathcal{C}}(y^{n}) = S(y^{n})/M$. Since $M = e^{nR}$ is deterministic,

\begin{equation}
  \E\!\left[P_{Y^{n}|\mathcal{C}}(y^{n})^{\beta}\right]=M^{-\beta} 
\E\left\{[S(y^{n})]^{\beta}\right\} = e^{-n\beta
R}\E\left\{[S(y^{n})]^{\beta}\right\}.
\end{equation}
The problem reduces to computing $\E[S(y^{n})^{\beta}]$.
For each $y^n$, 

\begin{eqnarray}
\E\{[S(y^n)]^\beta\}&=&\E\left[\sum_{\substack{Q_{XY}:\,Q_X=P_X \\
Q_Y=\hat{P}_{y^n}}}N(Q_{XY}|y^n)
\cdot e^{-n\ell(Q_{XY})}\right]^\beta\nonumber\\
&\doteq&\E\left[\sum_{\substack{Q_{XY}:\,Q_X=P_X \\
Q_Y=\hat{P}_{y^n}}}[N(Q_{XY}|y^n)]^\beta
\cdot e^{-n\beta\ell(Q_{XY})}\right]\nonumber\\
&=&\sum_{\substack{Q_{XY}:\,Q_X=P_X \\
Q_Y=\hat{P}_{y^n}}}\E\{[N(Q_{XY}|y^n)]^\beta\}
\cdot e^{-n\beta\ell(Q_{XY})}\nonumber\\
&\doteq&\max_{\substack{Q_{XY}:\,Q_X=P_X \\
Q_Y=\hat{P}_{y^n}}}\E\{[N(Q_{XY}|y^n)]^\beta\}
\cdot e^{-n\beta\ell(Q_{XY})},
\end{eqnarray}
where the dotted equalities are since the sum contains at 
most $\mathrm{poly}(n)$ non-negative terms.
Now, the evaluation of $\E\{N(Q_{XY}|y^n)]^\beta\}$ requires the following lemma,
which is Theorem 4.2 of \cite{merhav2025} and the proof is therein.

\begin{lemma}[Moments of Binomial Enumerator]
\label{lem:binomial-moments}
Let $N \sim \mathrm{Binomial}(e^{nA}, e^{-nB})$ with $A,B > 0$ and $\beta > 0$. Then,

\begin{equation}
\E\{N^{\beta}\}\doteq\left\{\begin{array}{ll}
e^{n\beta(A-B)} & A>B\\
e^{-n(B-A)} & A<B\end{array}\right.
\end{equation}
\end{lemma}

The case $A=B$ corresponds to a non-exponential behavior of $\E\{N^{\beta}\}$
as can be seen by observing the limiting behavior of both cases. (As a side
remark, when
$N$ is $\mathrm{Binomial}(e^{nA}, \lambda e^{-nA})$, for some constant
$\lambda>0$, then in the
limit of $n\to\infty$, $N$
becomes a Poissonian random variable with parameter $\lambda$, that is,
$\mbox{Pr}\{N=k\}\to\frac{\lambda^ke^{-\lambda}}{k!}$, for every non-negative
integer $k$, and therefore, the asymptotic $\beta$th moments are constants).

Lemma \ref{lem:binomial-moments} is now used with the assignments $A=R$ and $B=I_Q(X;Y)$.
Accordingly, one must distinguish between types $\{Q_{XY}\}$ 
for which  $R>I_Q(X;Y)$ as opposed to those with $R<I_Q(X;Y)$. Now,

\begin{eqnarray}
\E\{[S(y^n)]^\beta\}&\doteq&\max_{\substack{Q_{XY}:\,Q_X=P_X \\
Q_Y=\hat{P}_{y^n}}}\E\{[N(Q_{XY}|y^n)]^\beta\}
\cdot e^{-n\beta\ell(Q_{XY})}\nonumber\\
&=&\max\bigg\{\max_{\substack{Q_{XY}:\, Q_{X}=P_{X},~Q_Y=\hat{P}_{y^n}\\ I_{Q}(X;Y)\leq R}}
\E\{[N(Q_{XY}|y^n)]^\beta\}
\cdot e^{-n\beta\ell(Q_{XY})},\nonumber\\
& &\max_{\substack{Q_{XY}:\, Q_{X}=P_{X},~Q_Y=\hat{P}_{y^n}\\ I_{Q}(X;Y)\geq R}}
\E\{[N(Q_{XY}|y^n)]^\beta\}
\cdot e^{-n\beta\ell(Q_{XY})}\bigg\}\nonumber\\
&\doteq&\max\bigg\{\max_{\substack{Q_{XY}:\, Q_{X}=P_{X},~Q_Y=\hat{P}_{y^n}\\ I_Q(X;Y)\leq R}}
e^{n\beta[R-I_Q(X;Y)]}
\cdot e^{-n\beta\ell(Q_{XY})},\nonumber\\
& &\max_{\substack{Q_{XY}:\, Q_{X}=P_{X},~Q_Y=\hat{P}_{y^n}\\ I_{Q}(X;Y)\geq R}}
e^{-n[I_Q(X;Y)-R]}
\cdot e^{-n\beta\ell(Q_{XY})}\bigg\}.
\end{eqnarray}
Since the latter expression involves maximization over $\{Q_{XY}\}$ when
$Q_Y=\hat{P}_{y^n}$ is held fixed, it is clear that it depends on
$y^n$ only via $\hat{P}_{y^n}$. The number of $\{y^n\}$ with
$\hat{P}_{y^n}=Q_Y$ is of the exponential order of $e^{nH_Q(Y)}$. 
Multiplying the per-$y^n$ contributions by $e^{nH_Q(Y)}$ and by $e^{-n\beta R}$,
yields the total contribution of type $Q_Y$, and finally, summing over all $\{Q_Y\}$, which is
exponentially equivalent to maximizing over all $\{Q_Y\}$, yields 
the asserted expression of $\psi(\beta,R)$. The part pertaining to types
with $I_Q(X;Y)\le R$ is identified with $\psi_{\mbox{\tiny b}}(\beta,R)$, and the one 
with $I_Q(X;Y)> R$ is associated with $\psi_{\mbox{\tiny s}}(\beta,R)$.
\end{proof}

For future reference, we need the following definition.

\begin{definition}[Sparse-feasible and bulk-feasible types]
\label{def:feasibility}
For a given rate $R > 0$, a joint distribution $Q_{XY}$ with $Q_X = P_X$ is called:
\begin{itemize}[nosep]
  \item \emph{sparse-feasible} if $I_Q(X;Y) > R$, i.e., it satisfies the constraint
    of the sparse branch optimization;
  \item \emph{bulk-feasible} if $I_Q(X;Y) \leq R$, i.e., it satisfies the constraint
    of the bulk branch optimization.
\end{itemize}
Every $Q_{XY}$ is either sparse-feasible or bulk-feasible (or both, if $I_Q(X;Y)=R$).
\end{definition}

\subsection{The Sparse Branch: Formula and Phase Structure}
\label{subsec:sparse-branch}

The sparse branch of the annealed free energy,

\begin{equation}
  \psi_{\mbox{\tiny s}}(\beta, R)
  = R(1-\beta)+\sup_{\substack{Q_{XY}:\, Q_{X}=P_{X}\\ I_{Q}(X;Y)>R}}
  F(Q_{XY}),
  \label{eq:psi-sparse-def}
\end{equation}
captures the contribution of \emph{sparse-type} codewords: those whose joint type
$Q_{XY}$ with the output $y^n$ satisfies $I_Q(X;Y)>R$.
For fixed $y^n$, such codewords are rarely encountered in a typical codebook of the ensemble.

\begin{definition}[The annealed optimizer $Q_\beta^{\mbox{\tiny s}}$ 
and its mutual information $I^{\mbox{\tiny s}}(\beta)$]
\label{def:Qs}
For each $\beta > 0$, define

\begin{equation}
Q_{\beta}^{\mbox{\tiny s}}(y|x) =
\frac{[W(y|x)]^{\beta}}{\sum_{y'\in\mathcal{Y}}[W(y'|x)]^\beta} \quad x \in \mathcal{X}.
  \label{eq:Qs-def}
\end{equation}
\end{definition}
It can be readily shown that $Q_\beta^{\mbox{\tiny s}}$
maximizes
$F(Q_{XY})=H_Q(Y|X) - \beta\ell(Q_{XY})$ over all $Q_{XY}$ with $Q_X=P_X$.
Accordingly, define

\begin{equation}
  I^{\mbox{\tiny s}}(\beta) := I_{Q^{s}_{\beta}}(X;Y).
  \label{eq:Is-def}
\end{equation}
The curve $R = I^{\mbox{\tiny s}}(\beta)$ marks the point where
$Q_\beta^{\mbox{\tiny s}}$ crosses from
sparse-feasible ($I_{Q^{s}_\beta}(X;Y) > R$, and so, 
$Q_\beta^{\mbox{\tiny s}}$ optimizes the sparse branch)
to bulk-feasible ($I_{Q^{s}_\beta}(X;Y) \leq R$, so $Q_\beta^{\mbox{\tiny s}}$ 
optimizes the bulk branch).\\

\noindent
{\em Closed-form formula.}
By Lemma~\ref{lem:binomial-moments}, when $I_Q(X;Y)>R$ the $\beta$th moment of the
enumerator satisfies $\E\{[N(Q_{XY}|y^n)]^\beta\}\doteq e^{n[R-I_Q(X;Y)]}$
regardless of $\beta$.
The resulting annealed calculation gives

\begin{equation}
\label{eq:sparse-closed}
\psi_{\mbox{\tiny s}}(\beta, R)
=R(1-\beta) + C(\beta)
  \quad\text{whenever } R< I^{\mbox{\tiny s}}(\beta),
\end{equation}
where 

\begin{equation}
\label{Cbeta}
C(\beta) = \sum_{x\in\mathcal{X}}
P_X(x)\log\left(\sum_{y\in\mathcal{Y}}[W(y|x)]^\beta\right). 
\end{equation}
The formula ceases to hold when $R> I^{\mbox{\tiny s}}(\beta)$ as
$Q_\beta^{\mbox{\tiny s}}$
is then bulk-feasible ($I_{Q^{s}_\beta}(X;Y)\leq R$) and hence does not comply
with the rate constraint of the sparse branch.\\

\noindent
{\em Phase structure.}
The sparse branch has its own phase transition at $R=I^{s}(\beta)$:
\begin{itemize}[nosep]
  \item For $R < I^{\mbox{\tiny s}}(\beta)$: $\psi_{\mbox{\tiny s}}(\beta,R)
    = R(1-\beta)+C(\beta)$, the closed-form linear formula.
    The sparse branch is maximized by $Q^{s}_\beta$.
  \item For $R > I^{s}(\beta)$: $Q_\beta^{\mbox{\tiny s}}$ is not sparse-feasible;
    the closed-form $R(1-\beta)+C(\beta)$ does not apply and the sparse branch
    must be evaluated directly from \eqref{eq:psi-sparse-def}.
\end{itemize}

\noindent
{\em Operational meaning.}
The sparse branch governs the behavior of the ensemble average $\E[Z_n(\beta|\mathcal{C})]$
when the latter is inflated by atypical codebooks.
Specifically, when $\psi_{\mbox{\tiny s}}(\beta,R) > \psi_{\mbox{\tiny b}}(\beta,R)$ 
(which occurs for sufficiently small $R$,
as characterized by the phase boundary $R^\star(\beta)$ derived in
Theorem~\ref{thm:phase} below), the annealed free energy satisfies
$\psi(\beta,R) = \psi_{\mbox{\tiny s}}(\beta,R)$,
meaning that $\E[Z_n(\beta|\mathcal{C})]$ is dominated by a sub-exponential fraction
of codebooks with atypical codewords (those containing a sparse-type codeword
with $I_Q(X;Y) > R$).
A typical codebook has $\frac{1}{n}\log Z_n(\beta|\mathcal{C})\approx
\psi_{\mbox{\tiny b}}(\beta,R)$.

\subsection{The Bulk Branch: Operational Meaning}

We now motivate why the bulk branch $\psi_{\mbox{\tiny b}}(\beta,R)$ deserves
a separate study within the analysis of $\psi(\beta,R)$.
As said before,
the bulk branch, $\psi_{\mbox{\tiny b}}(\beta,R)$, is the component of $\psi(\beta,R)$ driven by
\emph{typical} codewords --- those whose joint type with the output satisfies
$I_Q(X;Y)\leq R$, i.e., codewords that are plausible given the code rate.
A typical random codebook contains no sparse-type codewords
(those with $I_Q(X;Y)>R$) with high probability,
so for a typical fixed codebook
$Z_n(\beta|\mathcal{C})\doteq e^{n\psi_{\mbox{\tiny b}}(\beta,R)}$ at the exponential scale.
The annealed free energy $\psi(\beta,R) = 
\max\{\psi_{\mbox{\tiny b}}(\beta,R),\psi_{\mbox{\tiny s}}(\beta,R)\}$
can exceed $\psi_{\mbox{\tiny b}}(\beta,R)$ (when the sparse branch dominates)
because it is inflated by the rare codebooks that happen to contain a sparse-type codeword;
In the language of statistical physics, $\psi_{\mbox{\tiny b}}(\beta,R)$ 
is the \emph{quenched} free energy, which is
the free energy of a typical random code, while $\psi(\beta,R)$ is the \emph{annealed}
free energy (the free energy of the disorder-averaged partition function).
To support this observation,
we now state the self-averaging property of $Z_n(\beta|\mathcal{C})$,
for $\beta\ge 1$ and $R\ge R^\star(\beta)$.
The proof appears in the Appendix.

\begin{theorem}[self-averaging and quenched free energy]
\label{thm:self-averaging}
For $\beta\geq 1$ and $R>R^\star(\beta)$:
\begin{enumerate}[label=(\roman*)]
  \item $\displaystyle\lim_{n\to\infty}\frac{1}{n}\,
        \E\bigl\{\log Z_n(\beta|\mathcal{C})\bigr\}=\psi_{\mbox{\tiny b}}(\beta,R).$
  \item $\displaystyle\frac{1}{n}\log Z_n(\beta|\mathcal{C})
        \xrightarrow{\;\mathrm{a.s.}\;}\psi_{\mbox{\tiny b}}(\beta,R).$
\end{enumerate}
\end{theorem}

The condition $R>R^\star(\beta)$ is precisely the regime (regions A and B in the
phase diagram) where the bulk branch dominates the annealed free energy:
$\psi(\beta,R)=\psi_{\mbox{\tiny b}}(\beta,R)$ (Theorem~\ref{thm:phase}).
This is needed for the upper bound in the proof: Markov's inequality on
$\E[Z_n]$ gives $(1/n)\log Z_n \leq \psi(\beta,R)$, and this equals
$\psi_{\mbox{\tiny b}}(\beta,R)$ only when $R>R^\star(\beta)$.
The lower bound of the proof holds under the weaker condition $R>I^{\mbox{\tiny
b}}(\beta)$, so region C ($I^{\mbox{\tiny b}}(\beta)<R<R^\star(\beta)$), where
a typical codebook still lives in the bulk phase but the annealed average is
inflated by rare sparse-type codebooks, remains open.
Whether the result also extends to the condensed phase $R\leq I^{\mbox{\tiny
b}}(\beta)$, or to $\beta<1$, likewise remains open.

By the definition of R\'{e}nyi entropy, $H_\beta(Y) = \frac{1}{1-\beta}\log\sum_y P(y)^\beta$,
so $\log Z_n(\beta|\mathcal{C}) = \log\sum_{y^n}P_{Y^n|\mathcal{C}}(y^n)^\beta
= (1-\beta)\,H_\beta(P_{Y^n|\mathcal{C}})$ exactly, for every codebook.
Since $\frac{1}{n}\log Z_n(\beta|\mathcal{C})\to\psi_{\mbox{\tiny b}}(\beta,R)$
a.s.\ (Theorem~\ref{thm:self-averaging}),
the bulk branch encodes the R\'{e}nyi entropy rate of the output mixture at order $\beta$:

\begin{equation}
  \psi_{\mbox{\tiny b}}(\beta, R) \doteq \frac{1-\beta}{n}\,H_\beta(P_{Y^n|\mathcal{C}})
  \label{eq:bulk-renyi}
\end{equation}
for a typical fixed codebook.

\subsection{The Bulk Branch: Formula and Phase Structure}

\begin{definition}[The bulk optimizer $Q_\beta^{\mbox{\tiny b}}$ and its
mutual information $I^{\mbox{\tiny} b}(\beta)$]
\label{def:Qb}
For each $\beta > 0$, define

\begin{equation}
  Q_{\beta}^{\mbox{\tiny b}} = \argmax_{\substack{Q_{XY}:\, Q_{X}=P_{X}}}
\bigl[(1-\beta)I_Q(X;Y) + F(Q_{XY})\bigr],
  \label{eq:Qb-def}
\end{equation}
the unconstrained maximizer of the bulk objective, and set

\begin{equation}
  I^{\mbox{\tiny b}}(\beta) := I_{Q_{\beta}^b}(X;Y).
  \label{eq:Ib-def}
\end{equation}
\end{definition}
Obviously, for $R\geq I^{\mbox{\tiny b}}(\beta)$, 

\begin{equation}
\psi_{\mbox{\tiny b}}(\beta,R)=
\psi_{\mbox{\tiny b,u}}(\beta)
  := \sup_{Q_{XY}:~Q_{X}=P_{X}} \bigl[(1-\beta)I_Q(X;Y) + F(Q_{XY})\bigr].
\end{equation}
The curve $R = I^{\mbox{\tiny b}}(\beta)$ is therefore the \emph{bulk phase boundary}:
for $R \geq I^{\mbox{\tiny b}}(\beta)$ the rate constraint $I_Q(X;Y)\leq R$ is inactive (bulk phase),
while for $R < I^{\mbox{\tiny b}}(\beta)$ it is active (condensed phase).

The next lemma establishes an inequality relation between $I^{\mbox{\tiny
b}}(\beta)$ and $I^{\mbox{\tiny s}}(\beta)$.

\begin{lemma}[Ordering of the two optimizers]
\label{lem:Ib-Is}
For $\beta \geq 1$:

\begin{equation}
  I^{\mbox{\tiny b}}(\beta) \leq I^{\mbox{\tiny s}}(\beta),
\end{equation}
with strict inequality for $\beta > 1$.
At $\beta = 1$, $I^{\mbox{\tiny b}}(1) = I^{\mbox{\tiny s}}(1) = I(X;Y)$.
\end{lemma}

\begin{proof}
By the optimality of $Q_\beta^{\mbox{\tiny b}}$ in the bulk objective:

\begin{eqnarray}
\label{inequalityneededlater}
F(Q_\beta^{\mbox{\tiny b}})+(1-\beta)I^{\mbox{\tiny b}}(\beta)&\geq&
F(Q_\beta^{\mbox{\tiny s}})+(1-\beta)I^{\mbox{\tiny s}}(\beta)\nonumber\\
&=&C(\beta)+(1-\beta)I^{\mbox{\tiny s}}(\beta),
\end{eqnarray}
which for $\beta>1$ is equivalent to

\begin{equation}
I^{\mbox{\tiny s}}(\beta) - I^{\mbox{\tiny b}}(\beta)\geq\frac{C(\beta)-F(Q_\beta^{\mbox{\tiny b}})}{\beta-1} 
\end{equation}
where the right-hand side is clearly non-negative since the denominator
$\beta-1$ is positive and the numerator is non-negative as $C(\beta)$ is the
global maximum of $F(Q_{XY})$, which cannot be smaller than
$F(Q_\beta^{\mbox{\tiny b}})$. Therefore, $I^{\mbox{\tiny s}}(\beta) -
I^{\mbox{\tiny b}}(\beta)$ is non-negative as well.
\end{proof}

\subsection{Phase Structure of the Annealed Free Energy}

The annealed free energy formula (Theorem~\ref{thm:annealed}) holds for all $\beta > 0$.
However, the analysis of the phase boundary between the bulk and sparse branches is restricted
here to the range $\beta \geq 1$, for two reasons.
The first reason is motivational:
In the statistical-mechanics language, $\beta \geq 1$ is the low-temperature regime where
$Z_n(\beta)$ is closely related to R\'{e}nyi entropies
of order at least one.  This range captures the operationally most relevant quantities:
$\beta=1$ (Shannon entropy, resolvability threshold), $\beta = 2$ (collision entropy,
birthday attacks), $\beta\to\infty$ (minimum entropy,
$-\log\max_{y^n}P_{Y^n}(y^n)$), and
general $\beta > 1$ (guessing moments $\E[G(Y^n)^s] \doteq e^{nsH_{1/(1+s)}}$~\cite{arikan1996}.
The second reason is that the range $0< \beta<1$ does not appear to lend
itself to closed-form analysis, the main issue being the lack of guarantee
concerning the uniqueness of the
solution $R$ to the equation $\psi_{\mbox{\tiny
b}}(\beta,R)=\psi_{\mbox{\tiny s}}(\beta,R)$, which is the phase boundary of
the total annealed free energy, as will be defined next. In particular,
for $0 < \beta < 1$, the branches $\psi_{\mbox{\tiny b}}(\beta,R)$
and $\psi_{\mbox{\tiny s}}(\beta,R)$ may or may
not cross, and existence and uniqueness of a crossing are not established in general.

\begin{definition}[Annealed Phase Boundary]
\label{def:Rstar}
For $\beta \geq 1$, the \emph{annealed phase boundary} $R^\star(\beta)$ is the unique rate
at which the bulk and sparse branches are equal:

\begin{equation}
  \psi_{\mbox{\tiny b}}\bigl(\beta, R^\star(\beta)\bigr) = \psi_{\mbox{\tiny
s}}\bigl(\beta, R^\star(\beta)\bigr).
  \label{eq:Rstar-def}
\end{equation}
\end{definition}
Existence and uniqueness of $R^\star(\beta) \in (I^{\mbox{\tiny b}}(\beta),
I^{\mbox{\tiny s}}(\beta))$
are established in Theorem~\ref{thm:phase}(b) below, together with an explicit
formula.  

To facilitate the need to keep track of the various phase boundaries and the
quantities associated with them so far, the following table summarizes those
ingredients. 

\noindent\textbf{Summary of the two optimizers.}
\begin{center}
\renewcommand{\arraystretch}{1.4}
\begin{tabular}{lp{5.0cm}p{5.0cm}}
\toprule
 & $Q_\beta^{\mbox{\tiny b}}$ & $Q_\beta^{\mbox{\tiny s}}$ \\
\midrule
Definition & Unconstrained maximizer of $\psi_{\mbox{\tiny b}}(\beta,R)$ objective (eq.~\eqref{eq:Qb-def}) & 
Gibbs distribution: $Q_\beta^{\mbox{\tiny s}}(y|x) \propto [W(y|x)]^\beta$ per letter (eq.~\eqref{eq:Qs-def}) \\
Maximizes & $H_Q(Y) - \beta\bigl[I_Q(X;Y) + \ell(Q_{XY})\bigr]$ & $H_Q(Y|X) - \beta\ell(Q_{XY}) = F(Q_{XY})$ \\
MI at optimum & $I^{\mbox{\tiny b}}(\beta)$ & $I^{\mbox{\tiny s}}(\beta)$ \\
Governs & Bulk condensation boundary $R = I^{\mbox{\tiny b}}(\beta)$ &
Bulk/sparse crossover $R = R^\star(\beta)$ \\
At $\beta=1$ & True channel $W$, $I^{\mbox{\tiny b}}(1)=I(X;Y)$ & True channel
$W$, $I^{\mbox{\tiny s}}(1)=I(X;Y)$ \\
\bottomrule
\end{tabular}
\end{center}
\medskip

The following theorem provides a characterization of the phase structure of
the annealed free energy.

\begin{theorem}
\label{thm:phase}
For any DMC $W$, any $P_{X}$, and any $\beta \geq 1$, $R > 0$:
\begin{itemize}
  \item[(a)] For $R \leq I^{\mbox{\tiny s}}(\beta)$, $Q_\beta^{\mbox{\tiny s}}$ is sparse-feasible and

        \begin{equation}
          \psi_{\mbox{\tiny s}}(\beta, R) = R(1-\beta) + C(\beta).
          \label{eq:sparse-closed-thm}
        \end{equation}
  \item[(b)] (Annealed phase boundary, $\beta \geq 1$.) Let

        \begin{equation}
          R^\star(\beta) := \frac{C(\beta)-\psi_{\mbox{\tiny
b,u}}(\beta)}{\beta-1},
          \quad \beta > 1, \qquad R^\star(1) := I(X;Y).
          \label{eq:Rstar-def-formula}
        \end{equation}
        For all $\beta \ge 1$, $R^\star(\beta)\in[I^{\mbox{\tiny
b}}(\beta),I^{\mbox{\tiny s}}(\beta)]$, it is the unique rate at which
        $\psi_{\mbox{\tiny s}}(\beta, R) = \psi_{\mbox{\tiny b}}(\beta, R)$, and

\begin{equation}
\psi(\beta, R) = \begin{cases} \psi_{\mbox{\tiny s}}(\beta, R) & R <
R^\star(\beta),\\ \psi_{\mbox{\tiny b}}(\beta, R) & R \geq R^\star(\beta). \end{cases}
        \end{equation}
\end{itemize}
\end{theorem}
The curve $R = R^\star(\beta)$ is the \emph{annealed phase boundary}.

\begin{proof}

Part (a) was already shown earlier (but we include it here as part of
the theorem for completeness since it is used in part (b)):
Observe that
for $R \leq I^{\mbox{\tiny s}}(\beta)$, we have $I_{Q^{s}_\beta}(X;Y) \geq R$, 
so it is sparse-feasible. Since $Q_\beta^{\mbox{\tiny s}}$ achieves the global
maximum, $F(Q_\beta^{\mbox{\tiny s}}) = \sup_{I_Q(X;Y) > R} F(Q_{XY}) = C(\beta)$, 
and so, $\psi_{\mbox{\tiny s}}(\beta,R) = R(1-\beta) + C(\beta)$.

As for part (b), observe that
for $R \in (I^{b}(\beta), I^{s}(\beta))$, both simplified formulas hold simultaneously:
$\psi_{\mbox{\tiny b}}(\beta,R) = \psi_{\mbox{\tiny b,u}}(\beta)$,
since $R > I^{\mbox{\tiny b}}(\beta)$) and $\psi_{\mbox{\tiny s}}(\beta,R) = R(1-\beta)+C(\beta)$ (by part~(a),
since $R < I^{\mbox{\tiny s}}(\beta)$). Setting $\psi_{\mbox{\tiny
b}}(\beta,R) = \psi_{\mbox{\tiny s}}(\beta,R)$ and solving gives immediately

\begin{equation}
  R^\star(\beta) = \frac{C(\beta)-\psi_{\mbox{\tiny b,u}}(\beta)}{\beta-1}.
\end{equation}
Substituting $\psi_{\mbox{\tiny b,u}}(\beta) = F(Q_\beta^{\mbox{\tiny b}}) +
(1-\beta)I^{\mbox{\tiny b}}(\beta)$, we end up with:

\begin{equation}
  R^\star(\beta) = I^{\mbox{\tiny b}}(\beta) +
\frac{C(\beta)-F(Q_\beta^{\mbox{\tiny b}})}{\beta-1}.
  \label{eq:Rstar-split}
\end{equation}
Since $\beta-1 > 0$ and $C(\beta)\ge F(Q_\beta^{\mbox{\tiny b}})$, the second term
is non-negative, giving $R^\star(\beta) \geq I^{\mbox{\tiny b}}(\beta)$, with strict inequality since
$F(Q_\beta^{\mbox{\tiny b}}) < C(\beta)$ strictly for $\beta > 1$ (Lemma~\ref{lem:Ib-Is}).

For the upper bound, recall from the proof of Lemma~\ref{lem:Ib-Is} that

\begin{equation}
  C(\beta)-F(Q_\beta^{\mbox{\tiny b}})\leq (\beta-1)\bigl(I^{s}(\beta) - I^{b}(\beta)\bigr),
\end{equation}
which gives 

\begin{equation}
\frac{C(\beta)-F(Q_\beta^{\mbox{\tiny b}})}{\beta-1}
\leq I^{s}(\beta) - I^{b}(\beta).
\end{equation}
Substituting into~\eqref{eq:Rstar-split}: $R^\star(\beta) \leq I^{\mbox{\tiny
s}}(\beta)$,
again with strict inequality for $\beta > 1$.
Hence $R^\star(\beta) \in (I^{\mbox{\tiny b}}(\beta), I^{\mbox{\tiny
s}}(\beta))$, confirming the formula is valid.

To establish the uniqueness of the solution $R^\star(\beta)$ to the equation
$\psi_{\mbox{\tiny b}}(\beta,R)=
\psi_{\mbox{\tiny s}}(\beta,R)$, it remains to rule out any additional
solutions outside the interval $[I^{\mbox{\tiny b}}(\beta),
I^{\mbox{\tiny s}}(\beta)]$. 
For $R \leq I^{\mbox{\tiny b}}(\beta)$, 
$Q_\beta^{\mbox{\tiny s}}$ is still sparse-feasible (since $R \leq
I^{\mbox{\tiny b}}(\beta) < I^{\mbox{\tiny s}}(\beta)$),
so part~(a) gives $\psi_{\mbox{\tiny s}}(\beta,R) = R(1-\beta)+C(\beta)$.
The bulk constraint is active,
so $\psi_{\mbox{\tiny b}}(\beta,R) \leq \psi_{\mbox{\tiny b,u}}(\beta)$, hence:

\begin{eqnarray}
\Delta(\beta,R)&:=&\psi_{\mbox{\tiny s}}(\beta,R)-\psi_{\mbox{\tiny
b}}(\beta,R)\nonumber\\
&=&R(1-\beta)+C(\beta) - \psi_{\mbox{\tiny b}}(\beta,R)\nonumber\\
&\geq& R(1-\beta)+C(\beta) - \psi_{\mbox{\tiny b,u}}(\beta).
\end{eqnarray}
The right side is positive at
$R=I^{\mbox{\tiny b}}(\beta)$
(as shown above) and increasing as $R$ decreases (since $1-\beta < 0$).
Hence $\Delta(\beta,R) > 0$ for all $R \leq I^{\mbox{\tiny b}}(\beta)$.
For $R\ge I^{\mbox{\tiny s}}(\beta)$, consider the following.
Since $I^{\mbox{\tiny b}}(\beta)\le I^{\mbox{\tiny s}}(\beta)\le R$, 
the bulk constraint is inactive and
$\psi_{\mbox{\tiny b}}(\beta,R)=
\psi_{\mbox{\tiny b,u}}(\beta)$.
Since $R\ge 
I^{\mbox{\tiny s}}(\beta)>R^\star(\beta)$, and 
$\psi_{\mbox{\tiny b,u}}(\beta)=
R^\star(\beta)(1-\beta)+C(\beta)$
(from the formula of $R^\star(\beta)$), we have: 

\begin{equation}
\psi_{\mbox{\tiny b}}(\beta,R)=
\psi_{\mbox{\tiny b,u}}(\beta)=
R^\star(\beta)(1-\beta)+C(\beta)>R(1-\beta)+C(\beta)\ge
\psi_{\mbox{\tiny s}}(\beta,R),
\end{equation}
where the strict inequality uses 
$R>R^\star(\beta)$ and $1-\beta<0$,
and the last inequality is because $\psi_{\mbox{\tiny s}}(\beta,R) =
R(1-\beta)+\sup_{\{Q:~Q_X=P_X,~I_Q(X;Y)>R\}}F(Q_{XY})\le R(1-\beta)+C(\beta)$.
Hence $\Delta(\beta,R) < 0$ for all $R\ge I^{\mbox{\tiny s}}(\beta)$.
\end{proof}

\section{A Numerical Example}
\label{sec:zchannel}

In this section, we provide a numerical example and illustrate the behavior
the two branches of the annealed free energy as well as the phase boundary
curves. 

All phase diagrams are computed for an example of a Z-channel.
The details are as follows. The input and output alphabets are
$\mathcal{X}=\mathcal{Y}=\{0,1\}$ and the crossover probability is
$p = 0.45$, i.e.,
$W(0|0) = 1$, $W(1|0) = 0$, $W(0|1) = 0.45$, and $W(1|1) = 0.55$.
The input type is given by $P_{X} = (0.5, 0.5)$. Accordingly, the resulting
channel mutual information is $I(X;Y) = 0.2441$.

Figure~\ref{fig:phase} plots all phase boundaries together in the relevant
quadrant of the $(\beta, R)$ plane.

\begin{figure}[ht]
  \centering
  \includegraphics[width=0.82\textwidth]{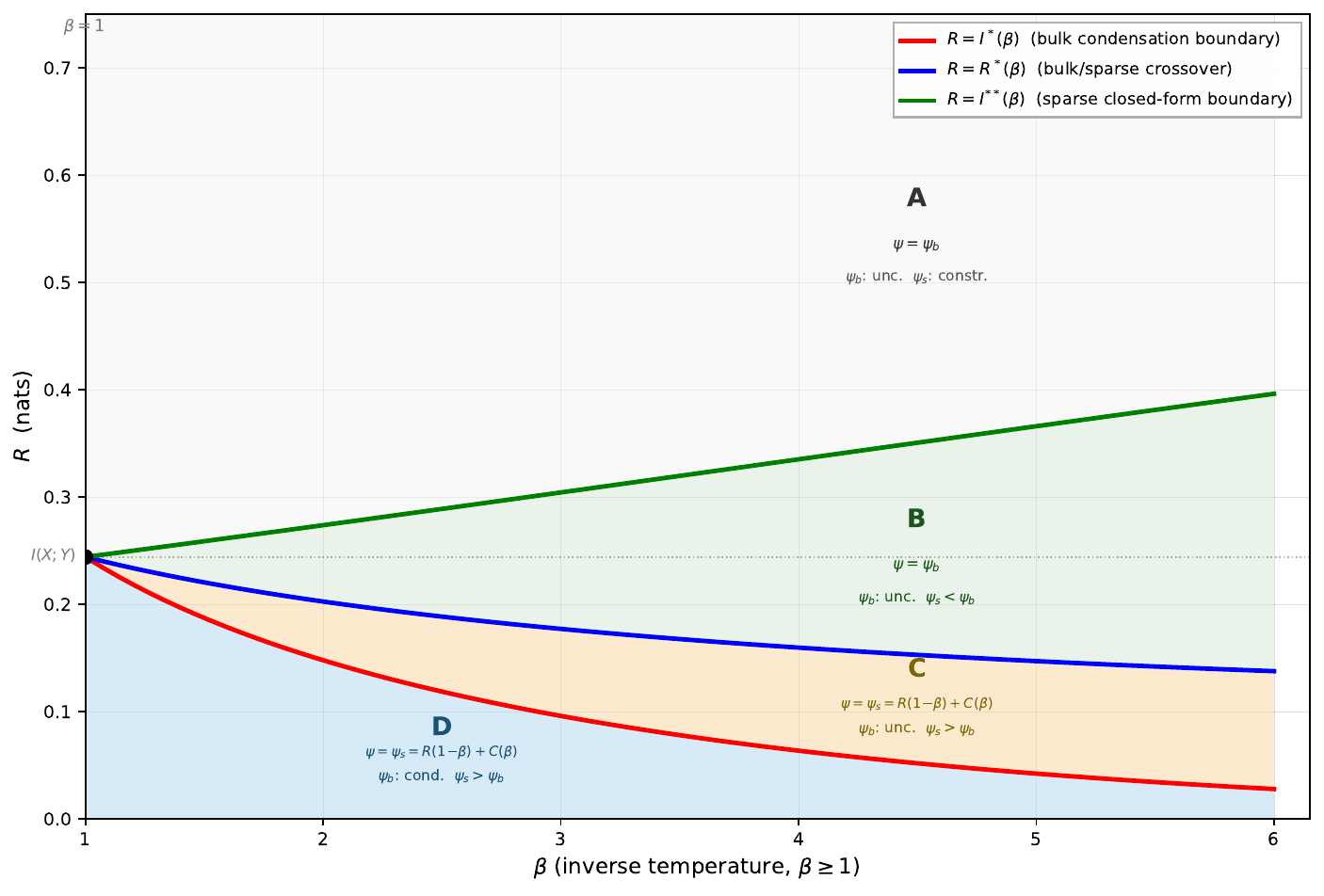}
  \caption{Phase boundaries for the Z-channel, $\beta \geq 1$.
    \textbf{Red:} $R = I^{\mbox{\tiny b}}(\beta)$, 
the bulk branch boundary (bulk rate constraint activates).
    \textbf{Blue:} $R = R^\star(\beta)$, the annealed boundary where bulk = sparse
    (closed-form formula, Theorem~\ref{thm:phase}(b)).
    \textbf{Green:} $R = I^{\mbox{\tiny s}}(\beta)$, 
the sparse branch boundary.
    All three meet at $(\beta,R)=(1,I(X;Y))$ (filled dot).
    Four regions (see text):
    \textbf{A} ($R > I^{\mbox{\tiny s}}(\beta)$): bulk and sparse both unconstrained;
    \textbf{B} ($R^\star(\beta) < R < I^{\mbox{\tiny s}}(\beta)$): 
sparse branch dominates ($\psi(\beta,R)=\psi_{\mbox{\tiny s}}(\beta,R)$);
    \textbf{C} ($I^{\mbox{\tiny b}}(\beta) < R < R^\star(\beta)$): 
bulk branch dominates ($\psi(\beta,R)=\psi_{\mbox{\tiny b}}(\beta,R)$);
    \textbf{D} ($R < I^{\mbox{\tiny b}}(\beta)$): bulk branch dominates, bulk condensed.}
  \label{fig:phase}
\end{figure}

Three phase boundary curves are visible in Figure~\ref{fig:phase}:
\begin{itemize}[nosep]
  \item $R=I^{\mbox{\tiny b}}(\beta)$ (red): the bulk branch boundary, where the unconstrained
    bulk optimizer crosses the rate constraint. At $\beta=1$: $I^{\mbox{\tiny
b}}(1)=I(X;Y)$.
  \item $R=R^\star(\beta)$ (blue): the annealed boundary where the bulk
    and sparse branches are equal, given by the closed-form
    formula~\eqref{eq:Rstar-def-formula}. At $\beta=1$: $R^\star(1)=I(X;Y)$.
    Decreasing in $\beta$, staying above $I^{b}(\beta)$.
  \item $R=I^{\mbox{\tiny s}}(\beta)$ (green): the sparse branch boundary, where 
    $Q_\beta^{\mbox{\tiny s}}$ (which maximizes $F(Q_{XY})=H_Q(Y|X)-\beta\ell(Q_{XY})$) has
    $I_{Q^{s}_\beta}(X;Y)=R$. For $R>I^{\mbox{\tiny s}}(\beta)$ the sparse closed-form ceases to hold.
    At $\beta=1$: $I^{\mbox{\tiny s}}(1)=I(X;Y)$. 
\end{itemize}
All three curves meet at the point $(\beta,R)=(1,I(X;Y))$, 
confirming that the soft-covering threshold
is the unique point where all regions collapse to a single point.

For $\beta \geq 1$, the three phase boundaries divide the quadrant
$\{\beta \geq 1, R > 0\}$ into four regions:\label{rem:three-regions}
\begin{itemize}[nosep]
  \item \textbf{A} ($R > I^{\mbox{\tiny s}}(\beta)$, above all curves):
    $\psi_{\mbox{\tiny b}}(\beta,R)$ is unconstrained ($I^{\mbox{\tiny b}}(\beta)<R$) and
    $\psi_{\mbox{\tiny s}}(\beta,R)$ is at its optimum ($I^{\mbox{\tiny
s}}(\beta)<R$).
    Both branches are unconstrained; $\psi(\beta,R)=\psi_{\mbox{\tiny b}}(\beta,R)$
    (bulk dominates since $R^\star(\beta)<I^{\mbox{\tiny s}}(\beta)$).
  \item \textbf{B} ($R^\star(\beta) < R < I^{\mbox{\tiny s}}(\beta)$):
    Sparse branch dominates: $\psi(\beta,R)=\psi_{\mbox{\tiny s}}(\beta,R)=R(1-\beta)+C(\beta)$
    (closed-form linear formula).
    Rare-event codewords inflate the ensemble average; a typical codebook
    is in the bulk phase.
  \item \textbf{C} ($I^{\mbox{\tiny b}}(\beta) < R < R^\star(\beta)$):
    Bulk branch dominates: $\psi(\beta,R)=\psi_{\mbox{\tiny b}}(\beta,R)$.
    The bulk rate constraint is inactive ($R>I^{\mbox{\tiny b}}(\beta)$); no condensation.
  \item \textbf{D} ($R < I^{\mbox{\tiny b}}(\beta)$, below the red curve):
    Bulk branch dominates with an active rate constraint: condensed phase.
    The bulk optimizer is constrained to $I_Q(X;Y)=R$.
\end{itemize}
For the very same z-channel example, figures~\ref{fig:fe-beta-R1},
\ref{fig:fe-beta-R2}, and
\ref{fig:fe-beta-R3}
display $\psi_{\mbox{\tiny b}}(\beta,R)$ and
$\psi_{\mbox{\tiny s}}(\beta,R)$
as functions of $\beta\geq 1$
for three representative values of $R$ (one below $I(X;Y)$, another one equal
to $I(X;Y)$, and yet another one above $I(X;Y)$)
and contrasts them with the
\emph{i.i.d.\ reference free energy}

\begin{equation}
  \psi_{\mathrm{iid}}(\beta)
  := \frac{1}{n}\log\left(\sum_{y^n} [P_{Y}(y^n)]^\beta\right)
  = \log\left(\sum_{y\in\mathcal{Y}} [P_Y(y)]^\beta\right),
  \label{eq:psi-iid}
\end{equation}
where $P_Y$ is the output distribution induced by $P_X$ and $W$.
This is the free energy of the product distribution $P_Y^{\otimes n}$,
i.e., the distribution to which $P_{Y^n}$ converges for $R>I(X;Y)$
based on well known soft-covering results. Clearly,
$\psi_{\mathrm{iid}}(\beta)$ is a purely smooth function with no phase
transitions whatsoever.

\begin{figure}[H]
\centering
\includegraphics[width=\textwidth]{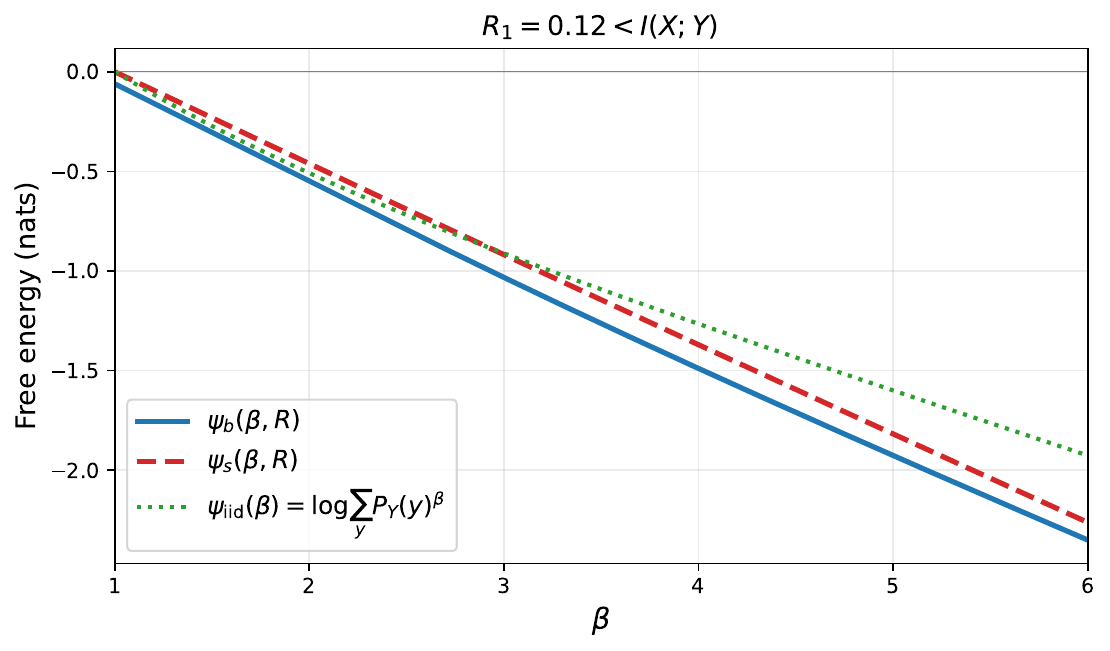}
\caption{$R_1=0.12 < I(X;Y)$. The sparse branch $\psi_{\mbox{\tiny s}}(\beta,R_1)$ lies above
the bulk branch for all $\beta\geq 1$, so the annealed free energy is dominated
by the sparse branch; both lie well above the i.i.d.\ reference,
reflecting ensemble inflation by rare codebooks.}
\label{fig:fe-beta-R1}
\end{figure}

\begin{figure}[H]
\centering
\includegraphics[width=\textwidth]{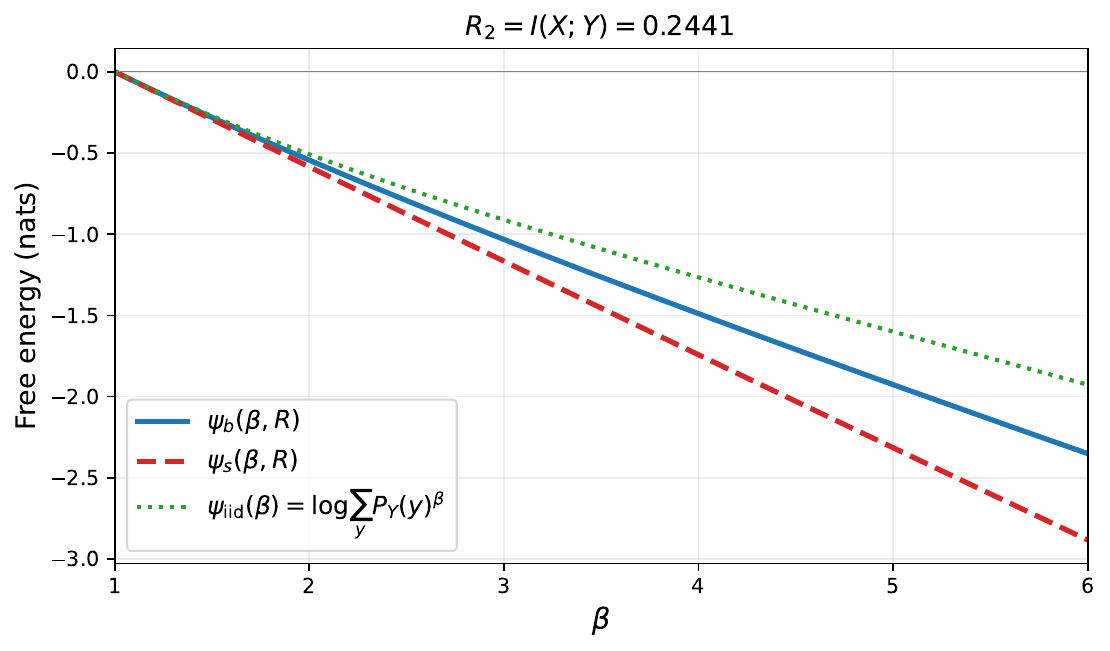}
\caption{$R_2 = I(X;Y) = 0.2441$ nats (the soft-covering threshold).
Both branches and the i.i.d.\ reference all vanish at $\beta=1$
(since $Z_n(1)=1$), and diverge differently for $\beta>1$.}
\label{fig:fe-beta-R2}
\end{figure}

\begin{figure}[H]
\centering
\includegraphics[width=\textwidth]{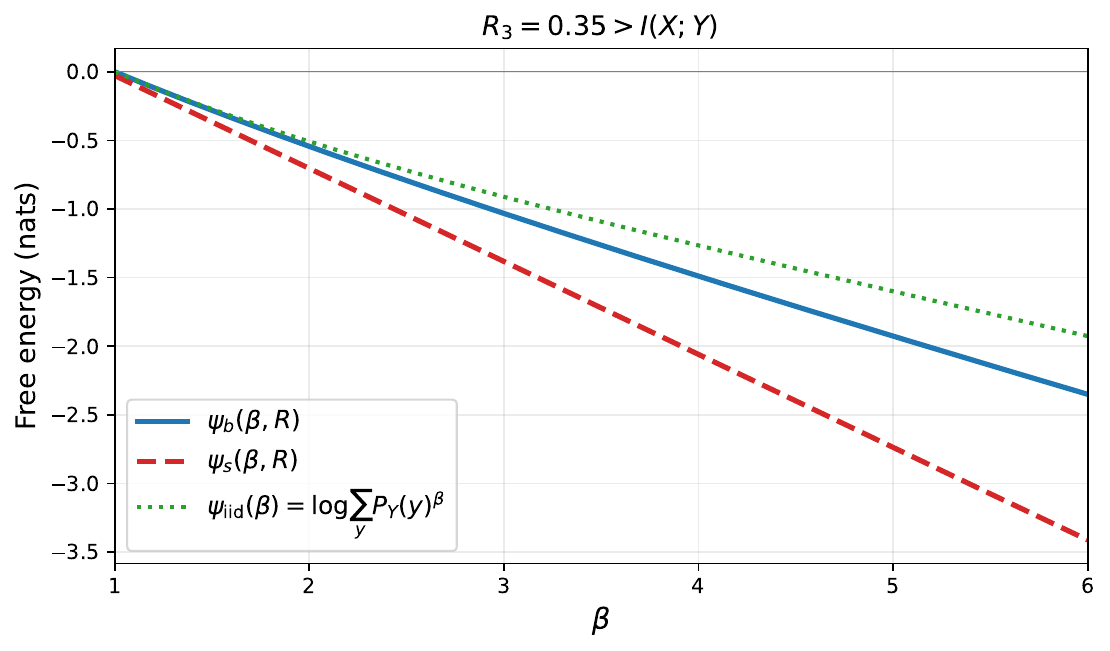}
\caption{$R_3=0.35 > I(X;Y)$. The bulk branch dominates throughout and tracks the
i.i.d.\ reference closely, confirming that above the soft-covering threshold
the output mixture approximates $P_Y^{\otimes n}$ in the R\'{e}nyi sense.
In all three figures, $\psi_{\mathrm{iid}}(\beta)\leq\psi(\beta,R)$,
with strict inequality below $I(X;Y)$: the random-coding mixture has strictly higher
R\'{e}nyi entropy than the i.i.d.\ output distribution at every order $\beta\geq 1$.}
\label{fig:fe-beta-R3}
\end{figure}

A central message of this paper is that the classical soft-covering threshold $R = I(X;Y)$ is
far from the end of the story, and as this numerical example illustrates, we
refine and deepen those results in this work. As mentioned in the
Introduction, classical results, like those of \cite{han1993}, \cite{hayashi2006}, \cite{yu2020} and
\cite{yu2019}, tell us that
$\|P_{Y^{n}} - P_{Y}^{\otimes n}\|_{\TV} \to 0$ (and in other metrics) if and only if $R > I(X;Y)$,
and provide exponential rates for this convergence.
Our statistical-physics analysis reveals that, even well above this threshold, $P_{Y^{n}}$
retains non-trivial internal structure that is invisible to total variation.
Figures~\ref{fig:fe-beta-R1},
\ref{fig:fe-beta-R2}, and
\ref{fig:fe-beta-R3}
make this visible. The three plots of $\psi_{\mbox{\tiny b}}(\beta,R)$,
$\psi_{\mbox{\tiny s}}(\beta,R)$, and the i.i.d.\ reference
$\psi_{\mathrm{iid}}(\beta) = \log\left(\sum_y [P_Y(y)]^\beta\right)$ against $\beta$ for $R_1 < I(X;Y) = R_2 < R_3$.
At the classical threshold $R_2 = I(X;Y)$ (second diagram), all three quantities
vanish together at $\beta = 1$ as required, but diverge for $\beta > 1$: the free energy $\psi(\beta, R_2)$
lies strictly above $\psi_{\mathrm{iid}}(\beta)$ for all $\beta > 1$,
meaning that $P_{Y^n}$ has strictly higher R\'{e}nyi entropy than $P_Y^{\otimes n}$ at every order.
Above the threshold at $R_3 > I(X;Y)$ (third diagram), the bulk branch closely
tracks $\psi_{\mathrm{iid}}(\beta)$, confirming the soft-covering intuition
that the output mixture approaches the i.i.d.\ law --- but the gap, while small,
remains nonzero at finite $\beta$, and quantifying it requires
$\psi_{\mbox{\tiny b}}(\beta, R_3)$.
Below the threshold at $R_1 < I(X;Y)$ (first diagram), the gap $\psi(\beta, R_1) - \psi_{\mathrm{iid}}(\beta)$
is large and grows with $\beta$: the sparse branch dominates the ensemble
average, and the random-coding output is far from i.i.d.\ in the R\'{e}nyi sense
even though TV-distance arguments do not apply here at all.

Specifically, $\psi_{\mbox{\tiny b}}(\beta, R)$ is the R\'{e}nyi entropy rate of $P_{Y^{n}}$ at order $\beta$ for
a typical codebook, and its phase transition at $R = I^{\mbox{\tiny b}}(\beta)$ implies a qualitative
change in how probability mass is distributed over output sequences. In the \emph{bulk phase}
($R > I^{\mbox{\tiny b}}(\beta)$), a typical output sequence is supported by exponentially many codewords---the
output distribution is diffuse and ``spread out'' in a way consistent with the intuition
behind soft-covering. In the \emph{condensed phase} ($R < I^{\mbox{\tiny b}}(\beta)$),
a typical output is supported by only sub-exponentially few codewords,
meaning the mass of $P_{Y^{n}}$ concentrates on a sparse set of sequences,
even though the TV distance to $P_{Y}^{\otimes n}$ may already be negligible.
This condensation is visible in Figure~\ref{fig:fe-beta-R1}: the first diagram ($R_1 < I(X;Y)$) shows
$\psi_{\mbox{\tiny s}}(\beta,R_1) > \psi_{\mbox{\tiny b}}(\beta,R_1)$ for all $\beta \geq 1$,
so the ensemble average is dominated by the rare sparse-branch codebooks,
while the third diagram ($R_3 > I(X;Y)$) shows $\psi_{\mbox{\tiny
b}}(\beta,R_3) > \psi_{\mbox{\tiny s}}(\beta,R_3)$
confirming typical-codebook dominance throughout.

The gap $\psi(\beta,R) - \psi_{\mathrm{iid}}(\beta) \geq 0$ visible in all
three diagrams has a direct operational meaning: it measures how much the
random-coding output mixture exceeds the i.i.d.\ distribution in
R\'{e}nyi entropy at order $\beta$. Classical soft-covering results establish that this
gap vanishes in total variation for $R > I(X;Y)$; The figures show that in R\'{e}nyi entropy
the gap is nonzero for every $\beta > 1$ and every finite $R$, decaying to zero only as $R \to \infty$.

\section{Applications and Implications}
\label{sec:applications}
In this section, we summarize several applications and implications of the
results in this work. Further investigation of these applications is deferred to future work.

\noindent
1.~{\em Refined output statistics.}
Resolvability shows $\|P_{Y^{n}} - (P_{Y})^{n}\|_{\TV} \to 0$ for $R > I(X;Y)$. Our phase diagram reveals that 
even in this regime, $P_{Y^{n}}$ may be in the condensed phase (below
$I^{\mbox{\tiny b}}(\beta)$ for large $\beta$), where mass concentrates on
outputs supported by a sub-exponential number of codewords. This is invisible
to the total variation distance and other distances.\\

\noindent
2.~{\em Guesswork.}
$\E[G(Y^{n})^{s}] \doteq e^{nsH_{1/(1+s)}(P_{Y^{n}})}$~\cite{arikan1996}.
Since $H_\beta(P_{Y^n|\mathcal{C}}) = \frac{1}{(1-\beta)n}\log Z_n(\beta|\mathcal{C})$
and $\frac{1}{n}\log Z_n(\beta|\mathcal{C}) \doteq \psi(\beta,R)$ (the full annealed free energy,
including both bulk and sparse branches), the guessing exponent at $\beta = 1/(1+s)$ is
$\frac{s}{n(1-\beta)}\log Z_n(\beta|\mathcal{C}) \doteq \frac{s}{1-\beta}\psi(\beta,R)$.
The phase structure of $\psi(\beta,R)$ --- both its bulk phase transition at
$R=I^{\mbox{\tiny b}}(\beta)$
and its bulk-sparse crossover at $R=R^\star(\beta)$ --- therefore produces regime changes
in guessing complexity absent in the i.i.d.\ case.\\

\noindent
3.~{\em Security.}
In the condensed phase, even when $P_{Y^{n}} \approx (P_{Y})^{n}$ in total
variation, only a sub-exponential number of codewords generate a typical output. 
A computationally powerful adversary can identify the message by examining which codewords are consistent 
with the observed output. Semantic security may therefore require operating strictly in the bulk phase 
for all $\beta$ relevant to the adversary's test.\\

\noindent
4.~{\em Hypothesis testing and the Chernoff exponent.}
A natural question is how well one can distinguish the random-code output mixture
$P_{Y^n|\mathcal{C}}$ from the i.i.d.\ reference $P_Y^{\otimes n}$.
With $P_Y=\mathrm{Unif}(\mathcal{Y})$ (achieved by choosing $P_X$ appropriately),
this is a direct test of whether soft covering has succeeded.
The annealed Chernoff exponent for this binary hypothesis test equals

\begin{equation}
  \xi(R) = \max_{0\leq\beta\leq 1}\bigl[(1-\beta)\log|\mathcal{Y}| - \psi(\beta,R)\bigr],
  \label{eq:chernoff-psi}
\end{equation}
where $\psi(\beta,R)$ is the annealed free energy of Theorem~\ref{thm:annealed}.
Since both branches of $\psi(\beta,R)$ are fully characterized by our results,
\eqref{eq:chernoff-psi} is completely determined.
The optimizer $\beta^\star\in[0,1]$ satisfies $\psi'(\beta^\star,R)=-\log|\mathcal{Y}|$,
so the phase structure of $\psi(\beta,R)$ directly governs the behavior of $\xi(R)$:
the bulk/sparse phase boundaries produce regime changes in the Chernoff exponent
as $R$ varies.

With $P_Y=\mathrm{Unif}$, two clean consequences follow.
First, $\xi(R)=0$ if and only if $R\geq I(X;Y)$: the two distributions are not
exponentially distinguishable for $R\geq I(X;Y)$, i.e., no test can achieve
an exponentially small total error probability, exactly at the soft-covering threshold.
(Sub-exponential error rates, e.g., polynomial in $n$, are not excluded.)
Second, \eqref{eq:chernoff-psi} is the Legendre--Fenchel transform of $\psi(\beta,R)$
with respect to $\beta$, evaluated at $\log|\mathcal{Y}|$; it therefore equals
$\frac{1}{n}\log\sum_{y^n}P_{Y^n|\mathcal{C}}(y^n)^{\beta^\star}|\mathcal{Y}|^{-n(1-\beta^\star)}$,
which is precisely the R\'{e}nyi divergence
$D_{\alpha}(P_{Y^n|\mathcal{C}}\,\|\,P_Y^{\otimes n})$ at order $\alpha=\beta^\star\in[0,1]$.
This connects \eqref{eq:chernoff-psi} directly to the R\'{e}nyi resolvability results of
Yu and Tan~\cite{yu2019}, who showed that
$D_\alpha(P_{Y^n|\mathcal{C}}\|P_Y^{\otimes n})\to 0$ if and only if $R>I(X;Y)$ for $\alpha\leq 1$,
while for $\alpha>1$ the threshold is $R=I^{\mbox{\tiny b}}(\alpha)$ --- the bulk
condensation boundary of our phase diagram.
Our analysis provides the exact \emph{rate} of divergence at every $R$ and $\alpha$,
and reveals how the two-branch phase structure of $\psi(\beta,R)$ determines
qualitative changes in this rate.

\section{Conclusion}
\label{sec:conclusion}

We developed a statistical-mechanical framework for the output distribution of random codes, 
centered on the annealed free energy $\psi(\beta,R)$ of the partition function 
$Z_{n}(\beta) = \sum_{y^{n}} [P_{Y^{n}}(y^{n})]^{\beta}$ and its two-branch phase structure.
This dichotomy between the two branches has concrete operational consequences. For guesswork, the phase structure of 
$\psi(\beta,R)$ --- both the bulk phase transition at $R=I^{\mbox{\tiny b}}(\beta)$ 
and the bulk-sparse crossover at $R=R^\star(\beta)$ --- translates directly 
(via Ar{\i}kan's formula~\cite{arikan1996}) into regime changes in the exponential 
growth rate of guessing moments $\E[G(Y^{n})^{s}]$. Similar comments apply to
hypothesis testing and perhaps other application areas.
The annealed free energy adds yet another layer: its phase boundary 
$R = R^\star(\beta)$ separates the regime where the ensemble average 
is dominated by rare, atypical codebooks (for example, the sparse branch dominates in
Figure~\ref{fig:fe-beta-R1}) from the regime where it reflects the behavior of
a typical codebook (the bulk branch dominates in Figure~\ref{fig:fe-beta-R3}). 
The gap $\psi(\beta,R) - \psi_{\mbox{\tiny b}}(\beta,R) > 0$ in the 
sparse-dominant regime quantifies the extent to which ensemble 
averages can be misleading---a warning particularly relevant when 
analyzing the security of specific random codes rather than random coding in the mean.

Future research directions include: a rigorous derivation of exponentially tight
error exponents for the hypothesis test of Application~4 via the type-class
enumeration method of~\cite{merhav2025}, together with a phase-diagram analysis
of the Neyman--Pearson exponent tradeoff in the $(R,\tau)$ plane;
a full treatment of the guesswork regime changes produced by the bulk/sparse
phase boundaries; quantitative semantic security bounds in terms of
$I^{\mbox{\tiny b}}(\beta)$ and $R^\star(\beta)$;
extension to Markov and mixed sources; connections to mismatched decoding exponents;
and a complete treatment of the $\beta<1$ regime.
Finally, we note that the self-averaging property of $Z_n(\beta|\mathcal{C})$
has been established in Theorem~\ref{thm:self-averaging}:
for $\beta\geq 1$ and
$R>R^\star(\beta)$, the bulk branch $\psi_{\mbox{\tiny b}}(\beta,R)$
coincides with the quenched free energy
$\lim_{n\to\infty}\frac{1}{n}\E\{\log Z_n(\beta|\mathcal{C})\}$.
Whether the same holds in region C ($I^{\mbox{\tiny b}}(\beta)<R<R^\star(\beta)$),
the condensed phase $R\leq I^{\mbox{\tiny b}}(\beta)$,
or for $\beta<1$, remains an interesting open problem.

\appendix
\section*{Appendix: Proof of Theorem~\ref{thm:self-averaging}}

Throughout this proof, we define 

\begin{equation}
U_n := \frac{1}{n}\log Z_n(\beta|\mathcal{C}).
\end{equation}

We begin with part (i).
To establish the upper bound, observe that by Jensen's inequality,
$\E\{U_n\}\leq\frac{1}{n}\log\E[Z_n(\beta|\mathcal{C})]\to\psi(\beta,R)=\psib$,
where the last equality is by the postulate $R\ge R^\star(\beta)$.
It follows that $\limsup_{n\to\infty}\E\{U_n\}\le\psib$.
As for the lower bound, fix any joint type $Q_{XY}$ with $Q_X=P_X$ and $I_Q(X;Y)\leq R$,
and let $Q_Y$ be its $Y$-marginal. For every output sequence $y^n$
of type $Q_Y$, each codeword of joint type $Q_{XY}$ with $y^n$ contributes
exactly $e^{-n\ell(Q_{XY})}$ to $\sum_m W^n(y^n|x^n(m))$, so

\begin{equation}
  P_{Y^n|\mathcal{C}}(y^n)
  \;\geq\; \frac{N(Q_{XY}|y^n)}{M}\cdot e^{-n\ell(Q_{XY})},
  \label{eq:app:Plb}
\end{equation}
and, since $\beta\geq 1$,

\begin{equation}
  Z_n(\beta|\mathcal{C})
  \;\geq\; \sum_{\substack{y^n:\,\hat{P}_{y^n}=Q_Y}}
    \!\!\left(\frac{N(Q_{XY}|y^n)}{M}\right)^{\!\beta}
    e^{-n\beta\ell(Q_{XY})}.
  \label{eq:app:Zlb}
\end{equation}
As observed before, $N(Q_{XY}|y^n)$ is binomial with $e^{nR}$ trials and a
probability of a single success of the exponential order of
$e^{-nI_Q(X;Y)}$, so its mean $\mu_n$ is exponentially
$e^{n(R-I_Q(X;Y))}\to\infty$.
By the Chernoff bound, for any $\delta\in(0,1)$:

\begin{equation}
  \Pr\left\{N(Q_{XY}|y^n)\leq(1-\delta)\mu_n\right\}
  \leq\exp\left(-\frac{\delta^2\mu_n}{2}\right)
  =\exp\left(-\frac{\delta^2}{2}e^{n[R-I_Q(X;Y)]}\right),
\end{equation}
which is \emph{doubly} exponentially small (see, also Theorem 4.1 of
\cite{merhav2025} and in
particular, eqs.\ (D.16), (D.17) in its proof therein, which establish the
doubly exponential decay).
A union bound over all
joint types $\{Q_{XY}\}$ and all
$|\mathcal{T}(Q_Y)|\leq e^{nH_Q(Y)}$ sequences $y^n$ gives

\begin{equation}
\Pr\{\mathcal{E}_n^{\mbox{\tiny c}}\}
\leq (n+1)^{|\mathcal{X}||\mathcal{Y}|}\cdot e^{nH_Q(Y)}
           \cdot\exp\left(-\tfrac{\delta^2}{2}\,e^{n[R-I_Q(X;Y)]}\right)
  \;\longrightarrow\; 0,
  \label{eq:app:union}
\end{equation}
where $\mathcal{E}_n$ is the event that $N(Q_{XY}|y^n)\geq(1-\delta)\mu_n$
for all types $Q_{XY}$ and all $y^n\in\mathcal{T}(Q_Y)$ simultaneously.
On $\mathcal{E}_n$, the bound~\eqref{eq:app:Zlb} gives

\begin{align}
Z_n(\beta|\mathcal{C})
&\geq |\mathcal{T}(Q_Y)|\cdot
\left(\frac{(1-\delta)\mu_n}{M}\right)^{\!\beta}
e^{-n\beta\ell(Q_{XY})} \notag\\
&\;\doteq\; (1-\delta)^\beta\cdot
\exp\!\Bigl(n\bigl[H_Q(Y)-\beta\bigl(I_Q(X;Y)+\ell(Q_{XY})\bigr)\bigr]\Bigr).
\label{eq:app:Zlb2}
\end{align}
Taking the supremum over all $Q_{XY}$ with $Q_X=P_X$ and $I_Q(X;Y)\leq R$,
we obtain

\begin{equation}
\label{lowerboundonEn}
U_n \;\geq\; \psib + \frac{\log(1-\delta)}{n}
\quad\text{on }\mathcal{E}_n.
\end{equation}
Taking expectations of both sides of \eqref{lowerboundonEn}
and using $U_n\geq(1-\beta)\log|\mathcal{Y}|$ deterministically
(since $Z_n(\beta)\geq|\mathcal{Y}|^{n(1-\beta)}$ by Jensen):

\begin{equation}
\E\{U_n\}
\geq\left(\psib+\frac{\log(1-\delta)}{n}\right)\Pr\{\mathcal{E}_n\}
  +(1-\beta)\log|\mathcal{Y}|\cdot\Pr\{\mathcal{E}_n^{\mbox{\tiny c}}\}.
\end{equation}
Since $\Pr[\mathcal{E}_n]\to 1$, we get $\liminf_{n\to\infty}\E\{U_n\}\geq\psib$,
which together with

\begin{equation}
\limsup_{n\to\infty}\E\{U_n\}\leq\psib, 
\end{equation}
implies

\begin{equation}
\lim_{n\to\infty}E\{U_n\}=\psib, 
\end{equation}
thus completing the proof of part (i).

Proceeding to part (ii),
we begin by proving that $\psib$ is an almost sure upper bound.
By Markov's inequality, for any $\varepsilon>0$:

\begin{equation}
\Pr\left\{U_n \geq \psib+\varepsilon\right\}
\leq \frac{\E[Z_n(\beta|\mathcal{C})]}{e^{n(\psib+\varepsilon)}}
\doteq e^{-n\varepsilon},
\end{equation}
where we have used the fact that $\frac{1}{n}\log\E[Z_n(\beta|\mathcal{C})]\to\psib$,
which holds because $\psi(\beta,R)=\psib$ for $R>R^\star(\beta)$.
Since $\sum_n e^{-n\varepsilon}<\infty$, the Borel--Cantelli lemma yields
$U_n\leq\psib+\varepsilon$ for every $\varepsilon>0$ eventually a.s., or
equivalently, $\limsup_{n\to\infty}U_n\le\psib$ a.s.
For the compatible almost sure lower bound, note that we have already proved
in (\ref{lowerboundonEn}) that whenever $\mathcal{E}_n$ occurs,
$U_n\geq\psib+\frac{\log(1-\delta)}{n}$, which is means that
$\mathcal{E}_n\subseteq \mathcal{F}_n :=\{U_n\geq\psib+\frac{\log(1-\delta)}{n}\}$, or
equivalently, $\mathcal{F}_n^{\mbox{\tiny c}}\subseteq
\mathcal{E}_n^{\mbox{\tiny c}}$. Thus,

\begin{eqnarray}
\sum_{n=1}^\infty\mbox{Pr}\left\{U_n<\psib+\frac{\log(1-\delta)}{n}\right\}=
\sum_{n=1}^\infty\mbox{Pr}\{\mathcal{F}_n^{\mbox{\tiny c}}\}\le
\sum_{n=1}^\infty\mbox{Pr}\{\mathcal{E}_n^{\mbox{\tiny c}}\}<\infty
\end{eqnarray}
where the summability of $\mbox{Pr}\{\mathcal{E}_n^{\mbox{\tiny c}}\}$ is due
to the fact that each term is bounded by a doubly exponentially small
quantity. By the Borel-Cantelli lemma, $U_n-\frac{\log(1-\delta)}{n}\ge\psib$
eventually a.s., or equivalently, $\liminf_{n\to\infty}U_n\ge\psib$ a.s., which
together with the matching upper bound yields $\lim_{n\to\infty}U_n=\psib$
a.s., completing the proof.

\end{document}